\documentclass[11pt,a4paper]{article}

\usepackage{geometry}
\usepackage{enumitem}
\usepackage{amssymb,amsthm,latexsym,amsmath,amscd,graphicx,url,xcolor,comment,mathrsfs}

\usepackage{hyperref}

\theoremstyle{definition}
\newtheorem{theorem}{Theorem}[section]
\newtheorem{proposition}[theorem]{Proposition}
\newtheorem{openproblem}[theorem]{Open Problem}
\newtheorem{lemma}[theorem]{Lemma}
\newtheorem{definition}[theorem]{Definition}
\newtheorem{example}[theorem]{Example}

\newtheorem{notation}[theorem]{Notation}
\newtheorem{corollary}[theorem]{Corollary}

\newtheorem{conjecture}[theorem]{Conjecture}

\newcommand{\N}{\mathbb N}

\newcommand{\Z}{\mathbb Z}

\newcommand{\F}{\mathbb F}

\newcommand{\Cc}{\mathcal C}

\newcommand{\fb}{\textnormal{FB}}
\newcommand{\fp}{\textnormal{FP}}

\usepackage{setspace}
\title{The Length of Functional Batch and PIR Codes}

\usepackage{authblk}

\author[1]{Altan B. K\i l\i\c{c}\thanks{A. B. K. is supported by the Dutch Research Council through grant VI.Vidi.203.045.}}

\author[2]{Alberto Ravagnani\thanks{A. R. is supported by the Dutch Research Council through grants VI.Vidi.203.045, 
OCENW.KLEIN.539, 
and by the Royal Academy of Arts and Sciences of the Netherlands.}}

\author[3]{Flavio Salizzoni\thanks{F. S. is supported by the P500PT\textunderscore  222344 SNSF project.}}

\affil[1,2]{
Eindhoven University of Technology, the Netherlands}

\affil[3]{Max Planck Institute for Mathematics in the Sciences, Germany}

\date{}

\begin{document}

\maketitle

\begin{abstract}
We consider the problem of computing the minimum length of functional batch and PIR codes of fixed dimension and for a fixed list size, over an arbitrary finite field. We generalize and refine several results that were previously obtained for binary codes. We present new upper and lower bounds for the minimum length and discuss the asymptotic behavior of this parameter, both for binary and non-binary fields.
We also compute its value for several parameter sets. Our approach also offers insights into the ``correct'' list size to consider for the Functional Batch Conjecture over non-binary finite fields, and establishes various supporting results.
\end{abstract}

\bigskip

\section{Introduction}

In the context of information retrieval,  querying a database can 
leak sensitive information
about the user's intent or interests, by revealing which specific items are being requested.
To address this fundamental privacy problem, that is, retrieving information without revealing what is being retrieved, Private Information Retrieval (PIR) protocols were introduced in~\cite{chor1998private}. Subsequently, a class of codes known as \textit{PIR codes} was introduced in \cite{fazeli2015pir} to support PIR protocols.

PIR codes form a special class of \textit{Batch codes}, introduced in \cite{ishai2004batch} and proposed as a tool for private information retrieval just like PIR codes. 
What distinguishes these two code types is the nature of the retrieval request and the flexibility of the recovery process.
PIR codes are specifically designed so that a single requested symbol can be recovered (recovery is formally defined later in Definition \ref{def:recovery}) in several mutually disjoint ways. By providing multiple independent ``recovery sets'' for the same piece of data, PIR codes allow a user to distribute queries across different servers or nodes. Because the server cannot determine which specific recovery set is being utilized for which symbol, the user's interest remains private. Batch codes
generalize the PIR concept by allowing several symbols (which may be distinct from one another) to be recovered simultaneously. PIR codes correspond to batch codes in the specific case where all requested symbols in the batch are the same. 


Both these code classes can be further generalized to \textit{functional} PIR and \textit{functional} batch codes \cite{zhang2020bounds}, where the goal is to retrieve functions of the stored data, rather than the original data. This paper focuses on the block length of functional PIR and batch codes. Functional PIR codes
are used to retrieve one function, while functional Batch
codes allow simultaneous retrievals of different functions.

Both functional batch and PIR codes are the subject of an intense research activity; see~\cite{chor1998private,fazeli2015pir,ishai2004batch,zhang2020bounds,hollmann2023some, lember2024equal, yohananov2021almost,nassar2021array, yohananov2025optimal} and the references therein, among many others. To our best knowledge, the vast majority of the research on PIR and batch codes focuses on binary codes.
Both for PIR and batch codes, the length $n$ is a crucial parameter measuring the number of storage nodes (or more generally memory units) required to accomplish a given storage/service task. More precisely, for a  given dimension $k$ and performance indicator~$t$, one would like to operate with a code having the smallest possible length $n_{\textnormal{min}}$. As it often happens in coding theory, these three parameters $(k,t,n)$ obey certain trade-offs, the simplest of which is
$$t\leq n_{\textnormal{min}}\leq kt.$$ 
None of these two bounds is met in general, and the picture is further complicated when field size $q$ is taken into account.

A central open problem in the theory of functional batch codes, which is relevant for this paper, is a conjecture by Zhang, Etzion, and Yaakobi~\cite[Conjecture 24]{zhang2020bounds}. It reads as follows.
\begin{conjecture}[{Functional Batch Conjecture}]
    The binary $[k,2^k-1]$ simplex code is a $2^{k-1}$-functional batch code.
\end{conjecture}

In \cite[Conjecture 2.5]{hollmann2023some}, Hollmann, Khathuria, Riet, and
Skachek conjectured that not only is the simplex code a $2^{k-1}$-functional batch code, but also that it has locality 2. Their conjecture coincides with the one proposed in a different context by Balister, Gy{\H{o}}ri, and Schelp.
\begin{conjecture}[see {\cite[Conjecture 1]{balister2011coloring}}]\label{conjecture2}
Given a list of $2^{k-1}$ nonzero vectors$v_1,\dots,v_{2^{k-1}}$ in~$\F_2^k$ such that $k\geq2$ and $\smash{\sum_{i=1}^{2^{k-1}}v_i=0}$, there exists a partition of $\F_2^k$ into 2-sets of vectors~$\{w_i, z_i\}$, for $i\in\{1,\dots,2^{k-1}\}$, such that $v_i =w_i+z_i$ for all $i$. 
\end{conjecture}
Partial results on this conjecture can be found in \cite{hollmann2023some,kovacs2023finding,yohananov2021almost}.
Even more generally, this conjecture can be stated as a specific case of a matching problem for finite abelian groups. For cyclic groups, the problem becomes the \textit{seating couple problem}; see for instance~\cite{karasev2012partitions,kohen2016generalization,kohen2010new,preissmann2009seating}. Conjecture \ref{conjecture2} is also related to \textit{rainbow matching}, and in this context Correia, Pokrovskiy, and Sudakov proved a result that implies the simplex code is a~$2^{k-1}-O(2^{k\frac{15}{16}})$-functional batch code for $k$ large, see~\cite[Theorem 1.2]{correia2023short}. In \cite{gao2021full}, Gao, Ramadurai, Wanless, and Wormald conjectured that a stronger version of this theorem may be true, and a positive answer to their conjecture would also imply a positive answer to the Functional Batch Conjecture.


In this paper, we consider functional PIR codes and functional batch codes over arbitrary finite fields, in sharp contrast with the majority of references on the subject. In Section 2, we establish the notation and  introduce two functions $\fp(k,t,q)$ and $\fb(k,t,q)$, which are the main objects investigated in this paper. 
These measure the minimum possible length of a functional PIR and a functional batch code, respectively, of given dimension $k$, performance metric $t$, and over the finite field $\F_q$.
In Section 3, we compute the exact value of these functions for several values of the parameters $k$, $t$, and $q$. Section 4 is devoted to upper and lower bounds, which generalize various of the results previously obtained for the binary field. Finally, in Section 5 we study the asymptotic behavior of $\fp(k,t,q)$ and $\fb(k,t,q)$ when $k$ is constant and $t$ goes to infinity and when $k$ and $t$ go to infinity simultaneously. 


\section{Problem statement and preliminaries}

We introduce the problem statement and establish the notation for the rest of the paper. 

\begin{notation}
In the sequel, $q$ denotes a prime power and $k,t,n \in \Z_{\ge 1}$. Let $M\in\F_q^{k\times n}$, $L$ be a list of nonzero vectors in $\F_q^k$, and $I \subseteq \{1,\ldots,n\}.$ Denote by $M^I$ the set of columns of~$M$ indexed by $I$. If $|I|=1$, say $I =\{i\}$, we simply write $M^i$ instead of $M^{\{i\}}$. Assume that the list $L$ consists of $\ell$ different vectors $v_1,\ldots,v_{\ell}$ such that they appear $m_1,\ldots,m_{\ell}$ times in $L$, respectively. In this case, we write 
$$L=\{v_1^{m_1},\ldots,v_{\ell}^{m_{\ell}}\}.$$
\end{notation}

Next, we formally define what it means for a matrix to ``serve'' a list of vectors.

\begin{definition}
\label{def:recovery}
Let $M\in\F_q^{k\times n}$ and $v \in \F_q^k$. A set $R \subseteq \{1,\ldots,n\}$ is called a \textbf{recovery set} for $v$ if  
$$v \in \langle M^{R} \rangle_{\F_q}.$$
A matrix $M\in\F_q^{k\times n}$ can \textbf{serve} the list $L = \{v_1,\dots, v_t\}$ of nonzero vectors in $\F_q^k$ if there exist pairwise disjoint
subsets $X_1,\dots,X_t\subseteq \{1,\ldots,n\}$ such that
$X_i$ is a recovery set for $v_i$
for all $i \in \{1,\ldots,t\}$. We equivalently say that $L$ is \textbf{served by} $M $. 
\end{definition}

Private information retrieval (PIR) codes have been studied extensively due to their property of having mutually disjoint recovery sets for any of the information bits. This concept has been extended to three other families of codes.

\begin{definition}
An $\F_q$-linear code $\Cc$ is: 
\begin{itemize}
    \item a \textbf{$t$-PIR} code if there exists a generator matrix $M\in\F_q^{k\times n}$ of $\Cc$ that can serve the list $L=\{e_i^t\}$ for all $i \in \{1,\ldots,k\}$;
    \item a \textbf{$t$-batch} code if there exists a generator matrix $M\in\F_q^{k\times n}$ of $\Cc$ that can serve any list of the form
    $L = \{e_1^{t_1},\ldots,e_k^{t_k}\}$ with $t_1+\dots+t_k=t$;
    \item a \textbf{$t$-functional PIR} code if there exists a generator matrix $M\in\F_q^{k\times n}$ of $\Cc$ that can serve the list $L=\{v^t\}$ for any nonzero vector $v \in \F_q^k$;
    \item a \textbf{$t$-functional batch} code if there exists a generator matrix $M\in\F_q^{k\times n}$ of $\Cc$ that can serve any list $L$ of $t$ nonzero vectors in $\F_q^k$.
\end{itemize}
\end{definition}
In this paper, we will only work on $t$-functional PIR and $t$-functional batch codes. Note that a $t$-functional batch is a $t$-batch code, and a $t$-functional PIR code is a $t$-PIR code. Moreover, a $t$-functional batch code is a $t$-functional PIR code. 

The main problem in studying these codes is to minimize their length given the values~$k$ and~$t$. That is, one wants to find the smallest integer $n$ for which there exist a $t$-functional PIR code and a $t$-functional batch code of dimension $k$ and length $n$. To this end, in this paper we study the following parameters.

\begin{definition}
    Let
\begin{align*}
    \fp(k,t,q)&=\min\{n \in \N \mid \mbox{there exists a }k\mbox{-dimensional }t\mbox{-functional PIR code}\},\\
    \fb(k,t,q)&=\min\{n \in \N \mid \mbox{there exists a }k\mbox{-dimensional }t\mbox{-functional batch code}\}.
\end{align*}
\end{definition}
\noindent It follows from the definitions that
\begin{equation*}
    \fb(k,t,q)\geq\fp(k,t,q).
\end{equation*}

The next observation shows that the choice of the generator matrix is irrelevant when computing the minima defined above.

\begin{proposition}
\label{prop:invertibleG}
Let $G \in \F_q^{k\times k}$ be any invertible matrix. We have that $M$ achieves $\fb(k,t,q)$, or $\fp(k,t,q)$, if and only if $GM$ does.
\end{proposition}

The next proposition establishes some general properties of the 
functions $\fp(k,t,q)$ and $\fb(k,t,q)$.

\begin{proposition}
\label{prop:increase+additive}
The functions $\fp(k,t,q)$ and $\fb(k,t,q)$ are:
\begin{enumerate}
    \item  strictly increasing in $k$ and $t$;
    \item  subadditive in $k$ and $t$.
\end{enumerate}
\end{proposition}
\begin{proof}
The subadditivity of the functions is proven in \cite{nassar2021array}, and the strict monotonicity of the functions in $t$ is proven in \cite{zhang2020bounds} for $q=2$. The same proofs also work for arbitrary field size. The fact that $\fp(k,t,q)$ and $\fb(k,t,q)$ are strictly increasing in $k$ will be proven in Proposition~\ref{proposition:incrbound}.
\end{proof}

We conclude this section by stressing that, to the best of our knowledge, the problem considered in this paper has only been studied for $q=2$. We refer to~\cite{yohananov2022almost,zhang2020bounds}.

\section{Exact values}
In this section, we compute the exact value of $\fb(k,t,q)$ for some choices of the parameters. When $k=1$, we clearly have $\fb(1,t,q)=t$ and $\fb(k,1,q)=k$. Therefore, we will focus only on $k \ge 2$.

\subsection{Length of 2-dimensional FB codes}

We focus on the case
$k=2$ for any prime power $q$. The final statement is the following.
\begin{theorem}
\label{thm:k=2}
$\fb(2,t,q)= \lceil 2(q+1)t/(q+2)\rceil = t +\lceil qt/(q+2)\rceil$.
\end{theorem}

To prove Theorem \ref{thm:k=2}, we need some preliminary results. The next lemma follows from a straightforward parity analysis.
\begin{lemma}
\label{lemma:easyfacts}
For any positive integer $z$, the following hold.
\begin{enumerate}
    \item $\lfloor (z+1)/2 \rfloor +1 = \lceil (z+2)/2 \rceil$,
    \item $\lceil (z+2)/2 \rceil + \lfloor z/2 \rfloor = z+1$,
    \item $\lfloor (z+1)/2 \rfloor = \left \lceil (z \lfloor (z+1)/2 \rfloor)/(z+2) \right \rceil = \left \lceil (z \lfloor (z+1)/2 \rfloor+z)/(z+2) \right \rceil$,
    \item $z = \left \lceil z(z+1)/(z+2)\right \rceil$.
\end{enumerate}

\end{lemma}
We continue with a lower bound on $\fb(2,t,q)$.
\begin{lemma}
\label{lem:lower_bound_k=2}
$\fb(2,t,q) \ge \left\lceil \frac{2(q+1)t}{q+2} \right\rceil = t + \left\lceil \frac{qt}{q+2} \right\rceil.$
\end{lemma}
\begin{proof}
Let $M$ be a matrix that attains $\fb(2,t,q)$. The multiset of columns of $M$ can be seen as a multiset of points in $PG(1,q).$ Let $v_0 = (0,1)^\top$ and $v_i = (1,i-1)^\top$ for $i \in \{1,\ldots,q\}$. Moreover, for $j \in \{0,\ldots,q\}$ let $n_j$ be the multiplicity of $v_j$ in the multiset of columns of $M$. 

If $t < n_j$ for some $j \in \{0,\ldots,q\}$, then $M$ cannot attain $\fb(2,t,q)$ since a matrix obtained by removing the extra $(n_j-t)$ columns that are equal to $v_j$ would suffice to serve any list of~$t$ nonzero vectors.
Therefore, we have $t \ge n_j$ for all $j \in \{0,\ldots,q\}$.

Let $j \in \{0,\ldots,q\}.$ To serve $t$ times the vector $v_j$, we must have 
\begin{equation}
\label{eq:referee2}
\fb(2,t,q) \ge n_j + 2(t-n_j).
\end{equation}

Equation \eqref{eq:referee2} follows from the fact that to serve $t$ copies of the vector $v_j$, we can use at most $n_j$ systematic positions, those are recovery set of size 1. After these are used, we must use larger recovery sets. In the most favorable case, these sets have size 2. Applying \eqref{eq:referee2} for all $j \in \{0,\ldots,q\}$ and summing the constraints yields
$$(q+1)\fb(2,t,q) \ge \fb(2,t,q) + 2(q+1)t - 2\fb(2,t,q),$$ which gives the desired result.
\end{proof}

The following lemma is proven by using the fact that any point of $PG(1,q)$ can be written as a linear combination of 2 different points of $PG(1,q).$

\begin{lemma}
\label{upper_bound_k=2}
For $1 \le t \le q+2$ we have 
$$\fb(2,t,q) \le \begin{cases}
        2t & \textnormal{ if } 1 \le t \le \lfloor (q+1)/2 \rfloor,\\
        2t-1 & \textnormal{ if } t =\lfloor (q+1)/2 \rfloor +1, \,\\
        2\lceil (q+2)/2 \rceil-1+2j & \textnormal{ if } t=\lceil (q+2)/2 \rceil + j \mbox{ for }j \in\{1,\ldots,\lfloor q/2 \rfloor\},\\
        2(t-1) & \textnormal{ if } t=q+2.
    \end{cases}$$
\end{lemma}
\begin{proof}
The proof constructs a matrix $M$ and distinguishes between the case $q+1$ odd and~$q+1$ even. We know that there are $q+1$ points in $PG(1,q)$, and that any such point can be written as a linear combination of 2 different points of $PG(1,q).$

\textbf{Case 1: }Let $q+1 = 2a$ for some positive integer $a\ge 2$. For each of the first $a$ vectors in a list of size $t$, we add 2 different points of $PG(1,q)$ as columns of $M$, so that after the $a$-th vector the set of columns of $M$ is equal to $PG(1,q).$ Then one only needs to add one column (for example~$e_1$) for the $(a+1)$-th vector, since one of the previously chosen $2a$ columns must be equal to that vector already. For the next $a-1$ vectors, we keep adding two different points of $PG(1,q)$ so that after the $(2a)$-th vector there are $2a+1 + 2(a-1)=4a-1$ columns in $M$. Each point of $PG(1,q)$ occurs twice as a column of $M$ except one point, call it $v$. For the $(2a+1)$-th vector (note that $2a+1=q+2$) we add the point $v$ as a column of $M$. It now follows that any list of $t \in \{1,\ldots,q+2\}$ vectors can be served with $M$. This is visually summarized in Table \ref{table:case1}.

\textbf{Case 2: }Let $q = 2b$ for some positive integer $b \ge 1$. The proof follows a similar strategy as Case 1 and thus omitted. However, it is summarized in Table \ref{table:case2}.

\begin{table}
\parbox{.49\linewidth}{
\centering
\resizebox{0.5\textwidth}{!}{
\begin{tabular}{|c| c|}
\hline
$t$& number of columns in $M$\\
\hline
 1  & 2  \\
\hline
2  & 4    \\
\hline
$\vdots$  & $\vdots$   \\
\hline
 $a$  & $2a = q+1$  \\
\hline
$a+1$  & $2a+1$   \\
\hline
$(a+1)+1$  & $(2a+1)+2$   \\
\hline
$\vdots$  & $\vdots$   \\
\hline
 $(a+1)+(a-1)=2a$  & $(2a+1)+2(a-1) = 4a-1$  \\
\hline
$2a+1 = q+2$  & $4a$   \\
\hline
\end{tabular}}
\caption{$q+1=2a$ is even.} \label{table:case1}
}
\hfill
\parbox{.49\linewidth}{
\centering
\resizebox{0.42\textwidth}{!}{
\begin{tabular}{|c| c|}
\hline
$t$& number of columns in $M$\\
\hline
 1  & 2  \\
\hline
2  & 4    \\
\hline
$\vdots$  & $\vdots$   \\
\hline
 $b$  & $2b$  \\
\hline
$b+1$  & $2b+1=q+1$   \\
\hline
$(b+1)+1$  & $(2b+1)+2$   \\
\hline
$\vdots$  & $\vdots$   \\
\hline
 $(b+1)+(b-1)$  & $(2b+1)+2(b-1)$  \\
\hline
$(b+1)+b = 2b+1$  & $(2b+1)+2b = 4b+1$   \\
\hline
$2b+2 = q+2$  & $4b+2$   \\
\hline
\end{tabular}}
\caption{$q+1=2b+1$ is odd.} \label{table:case2}
}
\end{table}

It can be seen from Tables~\ref{table:case1} and~\ref{table:case2} that when $1 \le t \le \lfloor (q+1)/2 \rfloor$, there are~$2t$ columns in $M$. If $t = \lfloor (q+1)/2 \rfloor + 1$, that is $t=a+1$ in Case 1 and $t=b+1$ in Case~2, there are~$2t-1$ columns in $M$. By combining the first equality in Lemma~\ref{lemma:easyfacts}, Table~\ref{table:case1}, and Table~\ref{table:case2} when $t = \lceil (q+2)/2 \rceil + j $ for some $j \in\{1,\ldots,\lfloor q/2 \rfloor\}$ there are $2t-1+2j$ columns in $M$. By the second equality in Lemma~\ref{lemma:easyfacts}, the last case to check is when $t=q+2$. Then there are $2(q+1) = 2(t-1)$ columns in~$M$, concluding the proof.
\end{proof}

We are now ready to prove Theorem \ref{thm:k=2}.
\begin{proof} [Proof of Theorem \ref{thm:k=2}]
We will prove the following statement equivalent to the theorem:
For $x,y \in \Z$ with $0\le x$ and $0 \le y < q+2$, we have 
$$\fb(2,x(q+2)+y,q) = 2(q+1)x + \left\lceil \frac{2(q+1)y}{q+2} \right\rceil.$$ By Lemma \ref{lem:lower_bound_k=2}, we have 
\begin{align*}
\fb(2,x(q+2)+y,q) &\ge x(q+2) + y + \left\lceil \frac{q(x(q+2)+y)}{q+2} \right\rceil  \\
&= x(q+2) + \left\lceil y+ \frac{(q+2)qx+qy}{q+2} \right\rceil \\
&= x(q+2)+qx + \left\lceil y+ \frac{qy}{q+2} \right\rceil \\
&= 2(q+1)x + \left\lceil \frac{2(q+1)y}{q+2} \right\rceil.
\end{align*}
By Proposition \ref{prop:increase+additive} we have 
$\fb(2,x(q+2)+y,q) \le x \fb(2,(q+2),q) + \fb(2,y,q).$ By Lemma~\ref{upper_bound_k=2}, we have $\fb(2,(q+2),q) = 2(q+1)$, which implies that 
$$\fb(2,x(q+2)+y,q) \le 2(q+1)x + \fb(2,y,q).$$ Note that if $y=0$ the lower and upper bounds coincide, giving the desired result. Therefore, for the rest of the proof we assume $1 \le y \le q+1$. 

By Lemma \ref{upper_bound_k=2}, we have 
\begin{equation}
\label{fb2yq}
\fb(2,y,q) \le \begin{cases}
        2y & \textnormal{ if } 1 \le y \le \lfloor (q+1)/2 \rfloor,\\
        2y-1 & \textnormal{ if } y =\lfloor (q+1)/2 \rfloor +1, \,\\
        2\lceil (q+2)/2 \rceil-1+2j & \textnormal{ if } y=\lceil (q+2)/2 \rceil + j \mbox{ for }j \in\{1,\ldots,\lfloor q/2 \rfloor\}.
    \end{cases}
\end{equation}
Therefore  to conclude the proof it suffices to show that
\begin{equation*}
\mbox{``the right-hand side of \eqref{fb2yq}''} = \left\lceil \frac{2(q+1)y}{q+2} \right\rceil = y + \left\lceil qy/(q+2) \right\rceil,
\end{equation*} so that the lower and the upper bound coincide. We analyze each case separately.
\begin{enumerate}
    \item Let $1 \le y \le \lfloor (q+1)/2 \rfloor$. If $y=1$, then $2y = y + \left\lceil qy/(q+2) \right\rceil$. If $y=\lfloor (q+1)/2 \rfloor$, then $2y = y + \left\lceil qy/(q+2) \right\rceil$ by the third equality in Lemma~\ref{lemma:easyfacts}.
    \item Let $y =\lfloor (q+1)/2 \rfloor +1.$ Then, $2y-1 = y + \left\lceil qy/(q+2) \right\rceil$ by the third equality in Lemma~\ref{lemma:easyfacts}.
    \item Let $y=\lceil (q+2)/2 \rceil + j \mbox{ for some }j \in\{1,\ldots,\lfloor q/2 \rfloor\}.$ If $y = \lceil (q+2)/2 \rceil + 1$, then $2\lceil (q+2)/2 \rceil-1+2 = y + \left\lceil qy/(q+2) \right\rceil$ by the first and the third equality in Lemma \ref{lemma:easyfacts}. Lastly, if $y = \lceil (q+2)/2 \rceil + \lfloor q/2 \rfloor$, then $2\lceil (q+2)/2 \rceil-1+2\lfloor q/2 \rfloor = y + \left\lceil qy/(q+2) \right\rceil$ by the second and the fourth equality in Lemma \ref{lemma:easyfacts}. \qedhere
\end{enumerate}
\end{proof}

When $q=2$, we obtain the following result as a corollary of Theorem \ref{thm:k=2}.

\begin{corollary}
\label{cor:(k,q)=(2,2)}
$\fb(2,t,2) = t + \lceil t/2 \rceil$. 
\end{corollary}
Note that the minimal length of a binary $t$-PIR code of dimension 2 coincides with the value of Corollary \ref{cor:(k,q)=(2,2)}; see e.g. \cite{fazeli2015pir}). We remark that this equality is only true for codes of dimension 2 as two independent vectors will not be, in general, enough to recover another vector when $k \ge 3.$

\subsection{Explicit computations for other special parameters}

In this subsection, we compute the values $\fp(k,t,q)$ and $\fb(k,t,q)$ for some specific choices of the triple $(k,t,q)$.
In \cite[Theorem 7]{zhang2020bounds}, it was shown that $\fp(k,2^{k-1},2)=2^k-1$, while in \cite[Theorem 4]{yohananov2022almost} the authors showed that $\fb(k,2^k,2)=2^{k+1}-2$. These two results can be generalized as follows.

\begin{lemma}\label{lemma:gencostr}
For every $s$ and $k\in\N$, we have that
\begin{equation*}
    \fp(k,2^{k-1}s,2)=(2^k-1)s \quad \text{ and } \quad \fb(k,2^{k}s,2)=(2^{k+1}-2)s.
\end{equation*}
\end{lemma}
\begin{proof}
    Since $\fp(k,2^{k-1},2)=(2^k-1)$, by the subadditivity of Proposition \ref{prop:increase+additive}, we immediately obtain that $$\fp(k,2^{k-1}s,2)\leq(2^k-1)s.$$
    Let $M\in\F_2^{k\times\fp(k,2^{k-1}s,2)}$ be the generator matrix of a $2^{k-1}s$-functional PIR code. Then, by the pigeonhole principle there exists a vector $v\in \F_2^k$ that appears $\ell$ times, with $$\ell\leq\lfloor \fp(k,2^{k-1}s,2)/(2^k-1)\rfloor,$$ among the columns of $M$. Since $M$ can serve $2^{k-1}s$ times the vector $v$ we obtain the following inequality
    \begin{equation*}
        \fp(k,2^{k-1}s,2)\geq  2(2^{k-1}s-\ell)+\ell =2(2^{k-1}s)-\ell\geq 2^ks-\frac{\fp(k,2^{k-1}s,2)}{2^k-1}.
    \end{equation*}
    Solving this inequality, we obtain $\fp(k,2^{k-1}s,2)\geq (2^k-1)s$, and this concludes the proof of the first equality. Starting from $\fb(k,2^{k},2)=(2^{k+1}-2)$ and proceeding in the same way, one can also prove the second inequality.
\end{proof}

Lemma \ref{lemma:gencostr} gives a curious corollary. Note that if the corollary were true for $t$ a multiple of $2^{k-1}$ instead of $2^k$, this would prove the Functional Batch Conjecture.
\begin{corollary}
\label{cor:multiple2k}
If $t$ is a multiple of $2^k$, then $\fp(k,t,2)=\fb(k,t,2)$.
\end{corollary}

Next, we show that the exact value of $\fp(k,t,2)$ can be computed recursively provided that one knows the value of $\fp(k,t,2)$ for some smaller values of $t$.

\begin{proposition}
\label{prop:recursion}
Let $s,h \in \N$ and $t=2^{k-1}s+h$ with $0\leq h< 2^{k-1}$. If $s+h\geq 2^{k}-1$. We have
\begin{equation*}
    \fp(k,t,2)= \fp(k,t-2^{k-1},2)+2^k-1. 
\end{equation*}
\end{proposition}
\begin{proof}
Let $M\in\F_2^{k\times\fp(k,t,2)}$ be the generator matrix of a $(2^{k-1}s+h)$-functional PIR code. If~$v\in\F_2^k$ is not a column of $M$, then in order to serve $t$ times the vector $v$ we  need at least~$2t$ columns. Thus
\begin{align*}
    \fp(k,t,2)&\geq 2(2^{k-1}s+h)= (2^ks+2h)=(2^k-1)s+(s+h) + h \geq (2^k-1)s+(s+h) \\
    &\ge (2^k-1)s+(2^k-1) = (2^k-1)(s+1) =
   \fp(k,2^{k-1}(s+1),2),
\end{align*}
where the last inequality follows from the assumption $s+h\geq 2^{k}-1$, and the last equality follows from Lemma~\ref{lemma:gencostr}. This contradicts Proposition~\ref{prop:increase+additive}, as $t<2^{k-1}(s+1)$ since $h<2^{k-1}$. Therefore, we can assume that every vector of $\F_2^k$ appears among the columns of $M$. 

Fix $v\in\F_2^k - \{0\}$. Let $X_1,\dots,X_t$ be a partition of the columns of $M$ that serves $t$ times the vector $v$. For every vector $w\in\F_2^k$ consider the recovery sets $X_i=\{w,v+w\}$ of $v$ for~$i \in \Z_{\ge 1}.$ Therefore, the matrix $M'$, which we obtain by removing from $M$ each vector of~$\F_2^k$ exactly once,  can still serve at least $t-2^{k-1}$ times $v$ (we subtract $2^{k-1}$ since there were~$2^{k-1} = 2^k/2$ different recovery sets of the form $X_i$ to begin with). Since this is true for every nonzero $v\in\F_2^k$, we obtain
\begin{equation*}
    \fp(k,t,2)\geq\fp(k,t-2^{k-1},2)+2^k-1.
\end{equation*}
We conclude using the subadditivity of $\fp(k,t,2)$. Since $\fp(k,2^{k-1},2)=2^k-1$ by \cite[Theorem 7]{zhang2020bounds}, the results follows from Proposition \ref{prop:increase+additive}.
\end{proof}

We illustrate Proposition \ref{prop:recursion} with an example.

\begin{example}
We will compute $\fp(3,19,2)$. Take $(s,h)=(4,3)$ in Proposition \ref{prop:recursion}. This gives
\begin{equation}
\label{eq:k3t19}
\fp(3,19,2) = \fp(3,15,2)+7.
\end{equation}
By \cite[Theorem 5]{zhang2020bounds} we know that $\fp(3,16,2)=\fp(3,15,2)+1$ and $\fp(3,16,2)=28$. Therefore, $\fp(3,19,2)=34$ by equation \eqref{eq:k3t19}.
\end{example}

One of the main contributions of this paper is to identify the natural generalization of $\fp(k,2^{k-1},2)=2^k-1=|\F_2^k-\{0\}|$ for $q\neq 2.$ The next result does that by finding the largest $t$ such that $\fp(k,t,q)=q^k-1=|\F_q^k-\{0\}|$.

\begin{lemma}
\label{lemma:constrq}
    $\fp\left(k,\frac{q^k+q-2}{2},q\right)=q^k-1.$
\end{lemma}
\begin{proof}
Let $M\in \F_q^{k\times (q^k-1)}$ be a matrix in which every nonzero vector of $\F_q^k$ appears as a column. Fix a vector $v\in\F_q^k$. We will now perform a parity analysis on $q$. 

If $q$ is even, then $M$ can serve $(q^k+q-2)/2$ times the vector $v$ with the partition given by all the sets of the form $\{0,\alpha v\}$ with $\alpha\in\F_q^*$, and all the sets of the form $\{w,v+w\}$ with~$w\in\F_q\setminus\langle v\rangle$. This works because $(v+w)+w=v$ when $q$ is even. We have 
\begin{equation}
\label{eq:qnot2}
(q-1)+\frac{q^k-q}{2} = \frac{q^k+q-2}{2},
\end{equation}
where $q-1$ is the number of $\alpha$'s, $q^k-q$ is the number of vectors outside the line $\langle v \rangle_q$, and the division by 2 is required to avoid overcounting the sets of the form $\{w,v+w\}$.

If $q$ is odd, then $M$ can serve $(q^k+q-2)/2$ times the vector $v$ with the partition given by all the sets of the form $\{0,\alpha v\}$ with $\alpha\in\F_q^*$, and all the sets of the form $\{v+w,v-w\}$ (this set would be a singleton if $q$ were even) with $w\in\F_q\setminus\langle v\rangle$. This is a partition since given~$w_1,w_2\in\F_q\setminus\langle v\rangle$ such that $\{v+w_1,v-w_1\}\cap\{v+w_2,v-w_2\}\neq\emptyset$, then either $w_1=w_2$ or $w_1=-w_2$, and in both cases we have $\{v+w_1,v-w_1\}=\{v+w_2,v-w_2\}$. For the counting, the same reasoning of equation \eqref{eq:qnot2} works, giving    \begin{equation}\label{equation:leftinequality}
        \fp\left(k,\frac{q^k+q-2}{2},q\right)\leq q^k-1.
    \end{equation}
    Let now $M$ be a matrix that realizes $\fp(k,(q^k+q-2)/2,q)$. By Equation \eqref{equation:leftinequality} we conclude that there exists a vector $v$ such that there are $h\leq q-1$ columns of $M$ that belong to $\langle v\rangle$. Then, we have that$$\frac{\fp\left(k,\frac{q^k+q-2}{2},q\right)-h}{2}+h\geq \frac{q^k+q-2}{2},$$ hence 
$$\fp\left(k,\frac{q^k+q-2}{2},q\right) \ge q^k+q-2-h \ge q^k-1,$$ yielding the desired result.  
\end{proof}

As explained before, taking $q=2$ in Lemma \ref{lemma:constrq} gives the known result 
$$\fp(k,2^k-1,2)=2^k-1.$$ We further generalize Lemma \ref{lemma:constrq} as follows, where by taking $s=1$ one recovers Lemma~\ref{lemma:constrq}. Note  however that the lemma is used in the proof.

\begin{theorem}\label{thm:gencostrq}
    For every $s \in \Z_{\ge1}$, we have
    $$\fp\left(k,\frac{q^k+q-2}{2}s,q\right)=s(q^k-1).$$
\end{theorem}
\begin{proof} 
    We proceed as in the proof of Lemma \ref{lemma:gencostr}. By Lemma~\ref{lemma:constrq} and by the subadditivity established in Proposition~\ref{prop:increase+additive}, we have $$\fp\left(k,\frac{q^k+q-2}{2}s,q\right)\leq s(q^k-1).$$
    Let $M\in\F_q^{k\times\fp(k,\frac{q^{k}+q-2}{2}s,q)}$ be the generator matrix of a $(q^k+q-2)s/2$-functional PIR code. By the pigeonhole principle there exists a nonzero vector $v\in \F_2^k$ for which there are 
    $$h\leq\frac{(q-1)\fp\left(k,\frac{q^k+q-2}{2}s,q\right)}{q^k-1}$$
    columns of $M$ that belong to $\langle v\rangle$. Since $M$ can serve $s(q^{k}+q-2)/2$ times the vector $v$, we obtain the following inequality:
    \begin{equation*}
        \frac{\fp(k,\frac{q^k+q-2}{2}s,q)-h}{2}+h\geq \frac{(q^k+q-2)s}{2}.
    \end{equation*}
    This implies $\fp(k,\frac{q^k+q-2}{2},2)\geq (q^k-1)s$, concluding the proof.
\end{proof}

We continue by recalling an equality established in~\cite[Section VI]{yohananov2022almost}:
\begin{equation}
\label{eq:nexttoFBC}
\fb(k,2^k,2)=2(2^k-1)
\end{equation}  This was also independently proven in \cite{hollmann2023some} using an algorithm by Hall, see~\cite{hall1952combinatorial}. Similarly to Lemma \ref{lemma:constrq} and Theorem \ref{thm:gencostrq}, one can try to generalize Equation \eqref{eq:nexttoFBC} to an arbitrary finite field. To this end, we propose the following open questions.

\begin{openproblem}\label{openproblem1}
Prove or disprove that $$\fb\left(k,q^k+q-2,q\right) = 2q^k-2.$$
\end{openproblem}
We also state a ``projectivized" version of the previous question.
\begin{openproblem}
\label{op:projectivebatch}
Prove or disprove that $$\fb\left(k,\frac{q^k-1}{q-1}+1,q\right) = 2\frac{q^k-1}{q-1}.$$
\end{openproblem}
Notice that a positive answer to the latter question would also result in a positive answer to the former one. In the next section, we will give some partial result on Open Problem~\ref{openproblem1}.

\section{Bounds and field size}
\label{sec:bound+q}
This section is devoted to bounds on $\fb(k,t,q)$ and to understand how $\fb(k,t,q)$ behaves with respect to $q$. We start with a trivial bound that will be useful later on.
\begin{lemma}\label{lemma:basicbound}
We have $t+k-1 \le \fp(k,t,q) \le \fb(k,t,q) \le tk.$
\end{lemma}
\begin{proof}
When $t=1$, one needs a full rank matrix. For the lower bound on $\fp(k,t,q)$, we know that $\fp(k,1,q) = k$. Since $\fp(k,t,q)$ is strictly increasing in $t$, we have the lower bound on $\fp(k,t,q)$. The upper bound on $\fb(k,t,q)$ is obtained considering a matrix that is formed by concatenating $t$ many identity matrices of size $k$.
\end{proof}
In Proposition \ref{prop:increase+additive} we mentioned that $\fp(k,t,q)$ is strictly increasing in $k$. The following proposition gives a lower bound on how much it increases at each step.
\begin{proposition}\label{proposition:incrbound}
    We have $$\left\lceil\frac{(q-1)\fb(k,t,q)}{q^k-1}\right\rceil\leq \fb(k,t,q)-\fb(k-1,t,q).$$
    The same inequality holds also for $\fp(k,t,q)$.
\end{proposition}
\begin{proof}
    Let $M$ achieve $n=\fb(k,t,q)$. There exists a one-dimensional space $V$ for which at least $\smash{\frac{n(q-1)}{q^k-1}}$ columns of $M$ belong to $V$. By Proposition \ref{prop:invertibleG}, we can assume without loss of generality that $V$ is generated by $e_1$. Then, the matrix $M'$, obtained from $M$ by deleting the first row and all the columns corresponding to vectors that belong to $V$, can serve any sequence of $t$ vectors in $\F_q^{k-1}$. Therefore, we obtain 
    \begin{equation*}
        \fb(k,t,q)-\frac{(q-1)\fb(k,t,q)}{q^k-1}\geq \fb(k-1,t,q),
    \end{equation*}
    and we conclude by noticing that $\fb(k,t,q)-\fb(k-1,t,q)$ is an integer. The same proof also works for $\fp(k,t,q).$    
\end{proof}

We now turn to the behavior of $\fb(k,t,q)$ with respect to $q$.
\begin{proposition}
    For all $s_1,s_2\in\N$ such that $s_1\mid s_2\in\N$, we have
$$\fb(k,t,q^{s_2})\leq\fb\left(k,\left(\frac{s_2}{s_1}\right)t,q^{s_1}\right)\leq\left(\frac{s_2}{s_1}\right)\fb(k,t,q^{s_2}).$$
The same inequalities hold for $\fp(k,t,q)$.
\end{proposition}
\begin{proof}
    Let $M$ be a matrix that achieves $\fb(k,t,q^{s_2})$. Fix a basis $\{\alpha_1,\dots,\alpha_{s_2/s_1}\}$ of $\F_q^{s_2}$ over~$\F_q^{s_1}$. For each column $g$ of $M$ we find $s_2/s_1$ vectors that correspond to the coefficients of $g$ with respect to our basis.
    The matrix obtained by collecting all these vectors has~$s_2/s_1\fb(k,t,q^{s_2})$ columns and can serve any list of $\frac{s_2}{s_1}t$ vectors over $\F_{q^{s_1}}$. Indeed, let $L=\{v_1,\dots, v_{(s_2/s_1)t}\}$ be a list of $(s_2/s_1)t$ vectors. We construct the following list 
    {\small$$\tilde L=\{\alpha_1 v_1+\dots+\alpha_{s_2/s_1}v_{s_2/s_1},\alpha_1 v_{(s_2/s_1)+1}+\dots+\alpha_{s_2/s_1}v_{\alpha_12(s_2/s_1)},\dots,v_{(s_2/s_1)t-s_2}+\dots+\alpha_{(s_2/s_1)t}v_{s_2/s_1}\}$$}
    of $t$ vectors over $\F_{q^{s_2}}$. Then, this list can served by $M$. Let $X_1,\dots, X_t$ be the sets used to serve $\bar L$. For each $i\in[t]$, consider the set $\tilde X_i$ given by the vectors that correspond to the coefficients of the elements in $X_i$ with respect to our basis. Then, the vectors~$\{v_{(i-1)(s_2/s_1)+1},\dots,v_{i(s_2/s_1)}\}$ can be served using the vectors in $\tilde X_i$.
    Therefore, we obtained the inequality on the right-hand side of the lemma. On the other hand, since we are able to serve $\frac{s_2}{s_1}t$ vectors over $\F_{q^{s_1}}$,  with the same matrix we can also serve $t$ vectors over~$\F_{q^{s_2}}$.
\end{proof}
We now provide more precise lower bounds for $\fp(k,t,q)$ and $\fb(k,t,q)$.
\begin{lemma}\label{lemma:lowbound}
    We have
    \begin{equation*}
        \fp(k,t,q)>\frac{kt}{\log_q(et)+1}-t.
    \end{equation*}

\end{lemma}
\begin{proof}
    Let $n=\fp(k,t,q)$. We define $s\in\N$ as
    $$s=\max\left\{i\in\N:\binom{it}{i}q^i<q^k\right\}.$$
    We first prove that $n>st$. Assume towards a contradiction that $n \le st$. Then
    $$\binom{n}{s}q^s \le \binom{st}{s}q^s < q^k.$$
    For a given matrix $M$ that realizes $\fp(k,t,q)$ we can find a vector $v\in\F_q^k$ that does not belong to any linear space generated by $s$ columns of $M$.
    Therefore, we would obtain $n\geq(s+1)t$, a contradiction. Thus, $n > st$. 
    By the definition of $s$ we have
    $$q^k \le \binom{(s+1)t}{s+1}q^{s+1} < \left(\frac{(s+1)te}{s+1}\right)^{s+1}q^{s+1},$$
    and therefore $(s+1)(\log_q(et)+1)>k$. We conclude by isolating $s$ and multiplying by $t$.
\end{proof}

The ideas used in the previous proof
allow us to establish the following result as well.

\begin{proposition}
\label{prop:bounds_on_q}
Let $1 \le s \le k-1$. If $q^{k-s} \ge \binom{t(s+1)}{s}$, then $\fb(k,t,q) \ge t(s+1)$.    
\end{proposition}
    \begin{proof}
        Suppose $\fb(k,t,q) \le t(s+1)-1$ and let $M$ achieve $\fb(k,t,q)$. The groups of $s$ columns of $M$ span at most
        $\binom{t(s+1)-1}{s}q^s$ vectors. By assumption, we have
        $$\binom{t(s+1)-1}{s} < \binom{t(s+1)}{s} \le q^{k-s}.$$
        This means that there exists a vector $v$ that needs at least $s+1$ columns to be recovered. But then $\fb(k,t,q)\ge t(s+1)$, a contradiction.
    \end{proof}

Observe that Proposition \ref{prop:bounds_on_q} when $s=k-1$ says that $\fb(k,t,q) \ge kt.$ However, we know by Lemma \ref{lemma:basicbound} that $\fb(k,t,q) \le kt$. Therefore, Proposition \ref{prop:bounds_on_q} gives 
the following corollary.

\begin{corollary}\label{corollary:lowb}
If $q \ge \binom{kt}{k-1}$, then $\fb(k,t,q) = kt$.
\end{corollary}

We continue with a bound on $\fb(k,t,q)$ that does not assume any constraints on the field size, in contrast to Proposition \ref{prop:bounds_on_q}.

\begin{theorem}\label{theorem:majorbound}
 We have that
\begin{equation*}
    (t(q-1)+1)^{\fb(k,t,q)}\geq (q^k-1)^t.
\end{equation*}
In particular, for $q\neq 2$ we have that
\begin{equation*}
    \fb(k,t,q)\geq \frac{tk}{\log_{q-1}(t(q-1)+1)}.
\end{equation*}
\end{theorem}

\begin{proof}
    Let $M$ be a matrix that achieves $\fb(k,t,q)$. We denote by $\mathcal{M}$ the set of matrices in $\smash{\F_q^{k\times t}}$ with all columns different from zero, and by $\mathcal{N}$ the set of matrices in $\smash{\F_q^{\fb(k,t,q)\times t}}$ whose rows are different from zero in at most one entry. We have that $\smash{\lvert \mathcal{M}\rvert=(q^k-1)^t}$ and~$\smash{\lvert \mathcal{N}\rvert=(t(q-1)+1)^{\fb(k,t,q)}}$. Since for every matrix $G\in\mathcal{M}$ there exists a matrix $N\in\mathcal{N}$ such that $MN=G$, we have
    \begin{equation*}
        (t(q-1)+1)^{\fb(k,t,q)}\geq (q^k-1)^t.
    \end{equation*}
    Moreover, when $q \neq 2$, by taking the logarithm on both sides we obtain
    \begin{equation*}
        \fb(k,t,q)\geq\frac{t\log_{q-1}(q^k-1)}{\log_{q-1}(t(q-1)+1)}\geq\frac{t\log_{q-1}((q-1)^k)}{\log_{q-1}(t(q-1)+1)}=\frac{tk}{\log_{q-1}(t(q-1)+1)},
    \end{equation*}
    which concludes the proof.
\end{proof}

We pause to examine a specific instance of the problem, namely the case of serving two binary vectors. Thus, we investigate $\fb(k,t,q)$ when $(t,q)=(2,2)$. We present a construction that can serve any two nonzero vectors in $\F_2^k$. We note that it is an open problem to compute
$\fb(k,2,q)$ for any alphabet size.

\begin{lemma}
\label{lem:upper_bound_t=2}
$\fb(k,2,2) \le \left\lceil \frac{3k}{2} \right\rceil.$
\end{lemma}
\begin{proof}
Let $m$ be a positive integer. The lemma is equivalent to show that $\fb(2m,2,2) \le 3m$ and $\fb(2m+1,2,2) \le 3m+2$. Consider the matrices
$$M_{\textnormal{even}} = \left( I_{2m} \mid r_1 \ldots r_m\right) \mbox{ and }M_{\textnormal{odd}} = \left( I_{2m+1} \mid r'_1 \ldots r'_m e_1\right),$$

where~$r_i = e_{2i-1}+e_{2i}$ and $r'_i = e_{2i}+e_{2i+1}$ for $i \in \{1,\ldots,m\}$. Observe that the sum of the columns in both matrices equal to the zero vector, proving that any vector can be served twice. Thus, we assume that $L =\{a,b\}$ is the list of vectors to be served where $a \neq b$. 

\textbf{Case 1: }Let $k=2m$. We will show that $M_{\textnormal{even}}$ can be used to serve any two vectors proving $\fb(2m,2,2) \le 3m$. Let $p_i = \{2i-1,2i\}$ for $i \in \{1,\ldots,m\}$. Define $$J_i =(|\sigma(a) \cap p_i|,|\sigma(b) \cap p_i|,|\sigma(a) \cap \sigma(b) \cap p_i|)$$ for $i \in \{1,\ldots,m\}$.
Fix $i \in \{1,\ldots,m\}$. We explain the recovery scheme in Table \ref{table:rec1}. This simply follows from the construction of the matrix $M_{\textnormal{even}}$ and the fact that the columns in the non-systematic part have pairwise disjoint supports. Thus, for any of the $m$ intervals of size one needs at most 3 columns. Thus $\fb(2m,2,2) \le 3m$.

\textbf{Case 2: }Let $k=2m+1$. We will show that $M_{\textnormal{odd}}$ can be used to serve any two vectors, proving that $\fb(2m+1,2,2) \le 3m+2$. Let $p'_i = \{2i,2i+1\}$ for $i \in \{1,\ldots,m\}$. Define $$J_i =(|\sigma(a) \cap p_i|,|\sigma(b) \cap p_i|,|\sigma(a) \cap \sigma(b) \cap p_i|)$$ for $i \in \{1,\ldots,m\}$. Fix $i \in \{1,\ldots,m\}$. We explain the recovery scheme in Table \ref{table:rec2}. This simply follows from the construction of the matrix $M_{\textnormal{odd}}$ and the fact that the columns in the non-systematic part have pairwise disjoint supports. Thus, for any of the $m$ intervals of size one needs at most 3 columns. The only thing left is to check the first coordinates. There, it is possible that $a_1=b_1=1$, forcing one to use $e_1$ both in the systematic and the non-sytematic part of $M_{\textnormal{odd}}$. Therefore, $\fb(2m+1,2,2) \le 3m+2$. \qedhere

\begin{table}
\parbox{.49\linewidth}{
\centering
\resizebox{0.5\textwidth}{!}{
\begin{tabular}{|c| c|}
\hline
$J_i$& number of columns used to recover $L$\\
\hline
 $(0,0,0)$  & none. \\
\hline
$(0,1,0)$  & 1 systematic col.    \\
\hline
$(0,2,0)$  & 1 non-systematic col.  \\
\hline
$(1,0,0)$  & 1 systematic col.  \\
\hline
$(1,1,0)$  & 2 systematic col.   \\
\hline
$(1,1,1)$  & 2 systematic col. and 1 non-systematic col.   \\
\hline
$(1,2,1)$  & 1 systematic col. and 1 non-systematic col.   \\
\hline
$(2,0,0)$  & 1 non-systematic col.   \\
\hline
$(2,1,1)$  & 1 systematic col. and 1 non-systematic col.   \\
\hline
$(2,2,2)$  & 2 systematic col. and 1 non-systematic col.  \\
\hline
\end{tabular}}
\caption{$k=2m$ is even.} \label{table:rec1}
}
\hfill
\parbox{.49\linewidth}{
\centering
\resizebox{0.5\textwidth}{!}{
\begin{tabular}{|c| c|}
\hline
$J_i$& number of columns used to recover $L$\\
\hline
 $(0,0,0)$  & none. \\
\hline
$(0,1,0)$  & 1 systematic col.    \\
\hline
$(0,2,0)$  & 1 non-systematic col.  \\
\hline
$(1,0,0)$  & 1 systematic col.  \\
\hline
$(1,1,0)$  & 2 systematic col.   \\
\hline
$(1,1,1)$  & 2 systematic col. and 1 non-systematic col.   \\
\hline
$(1,2,1)$  & 1 systematic col. and 1 non-systematic col.   \\
\hline
$(2,0,0)$  & 1 non-systematic col.   \\
\hline
$(2,1,1)$  & 1 systematic col. and 1 non-systematic col.   \\
\hline
$(2,2,2)$  & 2 systematic col. and 1 non-systematic col.  \\
\hline
\end{tabular}}
\caption{$k=2m+1$ is odd.} \label{table:rec2}
}
\end{table} 
\end{proof}

It is shown in \cite{zhang2020bounds} that $\fb(k, 2, 2)/k < 3/2$ when
$k$ is sufficiently large. However, the smallest $k$ such that $\fb(k, 2, 2) \le \left\lceil \frac{3k}{2} \right\rceil$ is not known.

By Lemma \ref{lem:upper_bound_t=2} and Proposition \ref{prop:increase+additive}, for $k$ and $t$ even we have $\fb(k,t,2)\leq\frac{3}{4}kt$. By Corollary \ref{corollary:lowb}, we know that the function $\fb(k,t,q)$ for large $q$ stabilizes at the value $kt$. However, we do not know if it is locally increasing. So, another open problem is to describe the behavior of $\fb(k,t,q)$ and $\fb(k,t,q)$ as $q$ varies.

Our next goal is to establish a bound related to Open Problem \ref{openproblem1}. Even though we are not able to give a complete answer, we will prove that the difference between~$\fb(k,q^k+q-2,q)$ and $2(q^k-1)$ is at most linear in $k$. We start by recalling the following result, which we will need in the proof of Lemma~\ref{lemma:hollmann}.

\begin{theorem}[see {\cite[Theorem 4.4]{hollmann2023some}}]\label{hollmann}
    Let $(A,+)$ be a finite abelian group of cardinality $n$, and let $a_1,\dots,a_n$ be
a list in $A$. There exists an ordering $g_1,\dots,g_n$ of the elements of $A$ such that $g_1+a_1,\dots,g_n+a_n$ is a permutation of the elements of $A$ if and only if $a_1+\dots+a_n=0$.
\end{theorem}
\begin{lemma}\label{lemma:hollmann}
    For every $k\in\N$ we have
    \begin{equation*}
        \fb(k,q^k,q)\leq 2(q^k-1).
    \end{equation*}
\end{lemma}
\begin{proof}
Let $A = \F_q^k$ and $n=q^k$. Let $G$ be a matrix in which every nonzero vector in $\F_q^k$ appears exactly twice as a column and consider a list $L=\{a_1,\dots,a_n\}$ of nonzero elements in $A$. If there exists $a \in A$ such that $a=a_i$ for every $i\in\{1,\dots,n\}$, then by Theorem~\ref{hollmann} there exists an ordering $g_1,\ldots,g_n$ of the element of $A$ such that $a+g_1,\dots,a+g_n$ is a permutation of the elements of $A$. Therefore, $G$ can serve the vector $a$ exactly $n$ times by considering the recovery sets $\{g_i,a+g_i\}$ for $i \in \{1,\ldots,n\}$. 

Now assume that the list has at least two different elements, without loss of generality say $a_{n-1} \neq a_n$. Define 
$x_L$ as the sum of the first $n-1$ elements in the list $L$. We start by showing that one can always assume $x_L \neq 0$. If $x_L=0$, then $x_L - a_{n-1} = -a_{n-1}$. Consider a new list $L'=\{a_1,\ldots,a_{n-2},a_n,a_{n-1}\}$. We have $x_{L'} \neq 0$, since $a_{n-1} \neq a_n$ by assumption. Therefore, for the remainder of the proof we assume that $x_L \neq 0$.

Consider the list $L'=\{a_1,\ldots,a_{n-1},-x_L\}$. By Theorem \ref{hollmann}, there exists an ordering $g_1,\dots,g_n$ of the elements of $A$ such that $g_1+a_1,\dots,g_n - x_L$ is a permutation of the elements of~$A$. Then $g_1+a_1+(x_L-g_n+a_n),\dots, g_n-x_L+(x_L-g_n+a_n)$ is a a permutation of the elements of $A$ as well. Therefore, $G$ can serve the list $L'$ with the sets 
\begin{align*}
X_1&=\{g_1+a_1+(x_L-g_n+a_n),g_1+(x_L-g_n+a_n)\}, \\
&\,  \; \vdots \\
X_{n-1}&=\{g_{n-1}+a_{n-1}+(x_L-g_n+a_n),g_{n-1}+(x_L-g_n+a_n)\}, \\
X_n&=\{g_n-x_L+(x_L-g_n+a_n),g_n+(x_L-g_n+a_n)\} =\{a_n,g_n+(x_L-g_n+a_n)\}.
\end{align*}
We conclude that $G$ can serve the list $L$ with the sets $X_1,\dots,X_{n-1},\{a_n\}$.
\end{proof}

Since the function $\fb(k,t,q)$ is subadditive in $t$, the previous lemma implies the following bound.

\begin{proposition}\label{proposition:conjbound}
$2(q^k-1)=\fp(k,q^k+q-2)\leq\fb(k,q^k+q-2,q)\leq 2(q^k-1)+(q-2)k.$
\end{proposition}

For large $t$ we can refine the previous result as follows.

\begin{proposition}
$2(q^k-1)\frac{q^{2k}+q-2}{q^k+q-2}\leq\fb(k,q^{2k}+q-2,q)\leq 2q^k(q^k-1).$
\end{proposition}
\begin{proof}
    Let $M_1$ be a matrix that achieves $\fp(k,q^k+q-2,q)$ and $M_2$ one that achieves  $\fb(k,q^{k}(q^{k}-1),q)$. Moreover, let $L=\{v_1^{m_1},\dots, v_{\ell}^{m_{\ell}}\}$ be a list of length $q^{2k}+q-2$. Note that $q^{2k}+q-2\geq\frac{q^k-1}{q-1}(q^k+q-2)$. This implies that in any list of length $q^{2k}+q-2$ there are at least $q^k+q-2$ elements belonging to the same one-dimensional space. So, without loss of generality, we assume $m_1\geq q^k+q-2$.
    Since    
    $$q^{2k}+q-2 = q^k(q^k-1)+q^k+q-2,$$
    the matrix $(M_1|M_2)$ can serve $L$ simply by serving $\{v_1^{m_1}\}$ using $M_1$ and~$\{v_2^{m_2},\dots, v_{\ell}^{m_{\ell}}\}$ using~$M_2$. Therefore
    by Lemma \ref{lemma:hollmann} and Theorem \ref{thm:gencostrq} for $s=2$ we have
    \begin{equation*}
        \fb(k,q^{2k}+q-2,q)\leq \fb(k,q^{k}(q^{k}-1),q)+\fp(k,q^k+q-2,q)\leq 2(q^k-1)(q^k-1)+2(q^k-1),
    \end{equation*}
proving the desired upper bound. 

To prove the lower bound, let $n=\fb(k,q^{2k}+q-2,q)$ and let $M$ be a matrix that achieves $n$. By the pigeonhole principle there exists at least one 1-dimensional space $V$ for which there are less than
    $$x = \left\lfloor\frac{n(q-1)}{q^k-1}\right\rfloor$$
    columns of $M$ belonging to $V$. Therefore, it must hold that
$$\fb(k,q^{2k}+q-2,2) \ge x+2(q^{2k}+q-2-x) = 2(q^{2k}+q-2)-x.$$ This translates into the following inequality:
    \begin{equation*}
        \frac{1}{2}\left(\frac{n(q^k+q-2)}{q^k-1}\right)=\frac{1}{2}\left(n+\frac{n(q-1)}{q^k-1}\right)\geq\frac{1}{2}\left(n+\left\lfloor\frac{n(q-1)}{q^k-1}\right\rfloor\right)\geq q^{2k}+q-2,
    \end{equation*}
    from which we obtain
    \begin{equation*}
        n\geq 2(q^k-1)\frac{q^{2k}+q-2}{q^k+q-2},
    \end{equation*}
    concluding the proof.
\end{proof}

It is natural to ask what happens if we switch the parameters $k$ and $t$. 
By observing the cases where $k$ and $t$ are very small, one may be tempted to conjecture that $\fb(k,t,q)=\fb(t,k,q)$. However, this is false in general.
In fact, the following holds.

\begin{corollary}
    For every $q$, there exist $k$ and $t$ such that $\fb(k,t,q)<\fb(t,k,q)$.
\end{corollary}
\begin{proof}
 By Proposition \ref{proposition:conjbound} we have $\fb(k,q^k+q-2,2)\leq 2(q^k-1)+(q-2)k$. By Lemma~\ref{lemma:lowbound} we have
    \begin{equation*}
        \fb(q^k+q-2,k,q)\geq \fp(q^k+q-2,k,q)>\frac{(q^k+q-2)k}{\log(ek)+1}-k.
    \end{equation*}
    This implies
    \begin{equation*}
        \begin{split}
            \fb(q^k+q-2,k,q)-\fb(k,q^k&+q-2,q) \\&\geq\frac{(q^k+q-2)k}{\log_q(ek)+1}-k-(2(q^k-1)+(q-2)k)\\
            &=\frac{(q^k+q-2)k-(2(q^k-1)+(q-1)k)(\log_q(ek)+1)}{\log(ek)+1},
        \end{split}
    \end{equation*}
    from which
    \begin{equation*}
    \liminf_{k\to\infty}(\fb(q^k+q-2,k,q)-\fb(k,q^k+q-2,q))\geq\liminf_{k\to\infty}\frac{q^kk-2q^k\log_q(ek)}{\log(ek)}=+\infty.
    \end{equation*}
    This establishes the corollary.
\end{proof}
\begin{example}  
It can be checked, for instance, that $\fb(1032,10,2)>\fb(10,1032,2)$.
\end{example}

We conclude this section with an upper bound on $\fb(k,t,q)$, which we will use
later in the proof of Theorem~\ref{theorem:main}.

\begin{lemma}\label{lemma:upperbound}
 We have
 \begin{equation*}
     \fb(k,t,q)\leq\left\lceil\frac{h}{s}\right\rceil\frac{2(s+q-2+\log_q(s))(s-1+(q-2)\log_q(s))}{\log_q(s)},
 \end{equation*}
 where $h=\max\{k,t\}$ and $s=\min\{k,t\}$.
\end{lemma}
\begin{proof}
 Let $g=\lceil \log_q s\rceil$. We have
    \begin{equation*}
    \begin{split}
        \fb(s,s,q)&\leq \fb\left(q^g+q-2,q^g+q-2,q\right)\leq\left\lceil \frac{q^g+q-2}{g}\right\rceil\fb(g,q^g+q-2,q)\\
        &\leq 2\left\lceil \frac{q^g+q-2}{g}\right\rceil\left(q^g-1+(q-2)g\right)\leq 2\left(\frac{q^g+q-2}{g}+1\right)\left(q^g-1+(q-2)g\right) \\
        &=\frac{2(s+q-2+\log_q(s))(s-1+(q-2)\log_q(s))}{\log_q(s)},
    \end{split}
    \end{equation*} 
where the third inequality follows from Proposition \ref{proposition:conjbound}, and the last equality follows from the fact that $q^g=s$. We conclude by noticing that $\fb(k,t,q)\leq\left\lceil\frac{h}{s}\right\rceil\fb(s,s,q)$.
\end{proof}

\section{Asymptotic behaviour}

In Section~\ref{sec:bound+q} we discussed what happens for $q$ sufficiently large. For $q=2$, the behavior of $\fb(k,t,2)$ and $\fp(k,t,2)$ when $t$ is fixed has been studied in \cite{zhang2020bounds}. The goal of this section is to study the behavior of the functions $\fp(k,t,q)$ and $\fb(k,t,q)$ when $t$ goes to infinity. Both functions are strictly increasing in $t$ and are therefore not bounded from above. It is natural to ask
how quickly they approach infinity. Since the two functions are subadditive in $t$, by Fekete's Lemma the limit of the ratio exists and it coincides with the infimum, i.e.,
\begin{equation*}
    \lim_{t\to\infty} \frac{\fp(k,t,q)}{t}=\inf_{t\in\N}\frac{\fp(k,t,q)}{t}\quad\text{ and }\quad\lim_{t\to\infty} \frac{\fb(k,t,q)}{t}=\inf_{t\in\N}\frac{\fb(k,t,q)}{t}.
\end{equation*}
We wish to explicitly compute these limits as $k$ and $q$ range over all possible values. The following proposition addresses the case of functional PIR codes.
\begin{proposition}\label{proposition:limfpq}
We have
    \begin{equation*}        \lim_{t\to\infty}\frac{\fp(k,t,q)}{t}=\frac{2(q^{k}-1)}{q^k+q-2}.
    \end{equation*}
\end{proposition}
\begin{proof}
    By Theorem \ref{thm:gencostrq} we have
    \begin{equation*}        \lim_{t\to\infty}\frac{\fp(k,t,q)}{t}=\lim_{s\to\infty}\frac{\fp\left(k,\frac{q^k+q-2}{2}s,q\right)}{\frac{q^k+q-2}{2}s}=\lim_{s\to\infty}\frac{s(q^k-1)}{\frac{q^k+q-2}{2}s}=\frac{2(q^k-1)}{q^k+q-2},
    \end{equation*}
    as desired.
\end{proof}

The previous proposition can be easily extended to the case of functional batch codes assuming a positive answer to Open Problem~\ref{op:projectivebatch}. However, with some extra work one can obtain the same result without that assumption.
\begin{theorem}\label{theorem:ttoinfinity}
We have
    \begin{equation*}
         \lim_{t\to\infty}\frac{\fb(k,t,q)}{t}=\lim_{t\to\infty}\frac{\fp(k,t,q)}{t}=\frac{2(q^{k}-1)}{q^k+q-2}.
    \end{equation*}
\end{theorem}
\begin{proof}
    Let $N=\frac{(q^k+q-2)}{2}$. Observed that if $t>N\frac{q^k-1}{q-1}$, then in any list of $t$ nonzero vectors there are at least $N$ of them that belong to the same one-dimensional vector space. Therefore, $\fb(k,t,q)\leq \fb(k,t-N,q)+\fp(k,N,q)$. By iterating this argument we obtain
    \begin{equation*}
        \fb\left(k,N\frac{q^k-1}{q-1}+sN,q\right)\leq\fb\left(k,N\frac{q^k-1}{q-1},q\right)+s\fp\left(k,N,q\right),
    \end{equation*}
    which implies that
    \begin{equation*}
        \begin{split}
            \lim_{t\to\infty}\frac{\fb(k,t,q)}{t}&\leq\lim_{s\to\infty}\frac{\fb\left(k,N\frac{q^k-1}{q-1}+sN,q\right)}{N\frac{q^k-1}{q-1}+sN}\leq\lim_{s\to\infty}\frac{\fb\left(k,N\frac{q^k-1}{q-1}+sN,q\right)}{sN}\\
            &\leq\lim_{s\to\infty} \frac{\fb\left(k,N\frac{q^k-1}{q-1},q\right)+s\fp\left(k,N,q\right)}{sN}=0+\frac{\fp(k,N,q)}{N}=\frac{2(q^{k}-1)}{q^k+q-2}.
        \end{split}
    \end{equation*}
    To conclude the proof, it suffices to use that $\fp(k,t,q)\leq\fb(k,t,q)$ and 
    apply Lemma~\ref{lemma:constrq} and Proposition~\ref{proposition:limfpq}.
\end{proof}

We now turn to the behavior of $\fp(k,t,q)$ and $\fb(k,t,q)$ when both $k$ and $t$ go to infinity simultaneously. Since $\fb(k,t,q)$ is submodular in both variables $k$ and $t$, one can prove that
\begin{equation*}        \lim_{k\to\infty}\frac{\fb(k,k,q)}{k^2}=\inf_{k\in\N}\frac{\fb(k,k,q)}{k^2};
\end{equation*}
 see for instance \cite{capobianco2008multidimensional}. Moreover, in analogy with 
 Theorem~\ref{theorem:ttoinfinity} one obtains
\begin{equation}\label{equation:limit}
     \lim_{k\to \infty}\frac{\fp(k,k,q)}{k^2}=\lim_{k\to \infty}\frac{\fb(k,k,q)}{k^2}=0.
\end{equation} 

Our goal is to extend this result to two arbitrary divergent sequences of natural numbers~$(k_n)_{n\in\N}$ and $(t_n)_{n\in\N}$.

\begin{theorem}\label{theorem:main}
    Let $(k_n)_{n\in\N}$ and $(t_n)_{n\in\N}$ be divergent sequences of natural numbers such that $k_n\geq t_n$ for all $n\in\N$. We have 
    \begin{equation*}        1\leq\liminf_{n\to\infty}\frac{\fp(k_n,t_n,q)\log_q(t_n)}{k_nt_n}\leq\limsup_{n\to\infty}\frac{\fb(k_n,t_n,q)\log_q(t_n)}{k_nt_n}<4.
    \end{equation*}
\end{theorem}
\begin{proof}
 The left-most inequality follows from Lemma \ref{lemma:lowbound}. By Lemma \ref{lemma:upperbound} we have
 \begin{equation*}
     \begin{split}
         \fb(k_n,t_n,q)&< 2\left(\frac{k_n}{t_n}+1\right)\frac{(t_n+q-2+\log_q(t_n))(t_n-1+(q-2)\log_q(t_n))}{\log_q(t_n)}\\
         &\leq 2\left(\frac{2k_n}{t_n}\right)\frac{(t_n+q-2+\log_q(t_n))(t_n-1+(q-2)\log_q(t_n))}{\log_q(t_n)}.
     \end{split}
 \end{equation*}
Therefore, 
 \begin{multline*}
\limsup_{n\to\infty}\frac{\fb(k_n,t_n,q)
\log_q(t_n)}{k_nt_n} \\ \leq \limsup_{n\to\infty} 2\left(\frac{2k_n}{t_n}\right)\frac{(t_n+q-2+\log_q(t_n))(t_n-1+(q-2)\log_q(t_n))}{k_nt_n}\\
   \leq\limsup_{n\to\infty} 4\frac{(t_n+q-2+\log_q(t_n))(t_n-1+(q-2)\log_q(t_n))}{t_n^2}=4,
 \end{multline*}
 concluding the proof.
\end{proof}

The upper bound of the previous theorem can be improved if one has more information about the ratio
between $(k_n)_{n\in\N}$ and $(t_n)_{n\in\N}$. For instance, when $k_n=t_n=n$ for all $n\in\N$, we obtain the following.

\begin{corollary}\label{corollary:limk^2=0}
    We have
    \begin{equation*}
        1\leq\liminf_{n\to\infty}\frac{\fp(n,n,q)\log_q(n)}{n^2}\leq\limsup_{n\to\infty}\frac{\fb(n,n,q)\log_q(n)}{n^2}\leq2.
    \end{equation*}
\end{corollary}
The upper bound of Corollary \ref{corollary:limk^2=0} follows
from the fact we can skip an 
approximation in the proof of Theorem~\ref{theorem:main}. More precisely, in the first line of the proof of Theorem \ref{theorem:main}, we have ``$2$'' instead of ``$2(k_n/t_n)+1$''.

We conclude the section with a generalization of Equation \eqref{equation:limit} to arbitrary divergent sequences.
\begin{proposition}
    Let $(k_n)_{n\in\N}$ and $(t_n)_{n\in\N}$ be divergent sequences of natural numbers. We have
    \begin{equation*}
        \lim_{n\to\infty}\frac{\fp(k_n,t_n,q)}{k_nt_n}=\lim_{n\to\infty}\frac{\fb(k_n,t_n,q)}{k_nt_n}=0.
    \end{equation*}
\end{proposition}
\begin{proof}
    Let $h_n=\max\{k_n,t_n\}$ and $s_n=\min\{k_n,t_n\}$ for all $n$. Since $\fb(k,t,q)$ is subadditive in both variables, we have 
    \begin{equation*}
        \lim_{n\to\infty}\frac{\fb(k_n,t_n,q)}{k_nt_n}\leq \lim_{n\to\infty}\left(\frac{h_n}{s_n}+1\right)\frac{\fb(s_n,s_n,q)}{h_ns_n}=\lim_{n\to\infty}\left(\frac{h_n+s_n}{h_n}\right)\frac{\fb(s_n,s_n,q)}{s_n^2}=0,
    \end{equation*}
    where the first inequality follows from the fact that $k_nt_n=h_ns_n$, and the last equality follows from Corollary \ref{corollary:limk^2=0}.
\end{proof}
\bibliographystyle{plain}	
\bibliography{bibliography}

\begin{thebibliography}{10}

\bibitem{balister2011coloring}
Paul~N Balister, E~Gy{\H{o}}ri, and Richard~H Schelp.
\newblock Coloring vertices and edges of a graph by nonempty subsets of a set.
\newblock {\em European Journal of Combinatorics}, 32(4):533--537, 2011.

\bibitem{capobianco2008multidimensional}
Silvio Capobianco.
\newblock Multidimensional cellular automata and generalization of {F}ekete's lemma.
\newblock {\em Discrete Mathematics \& Theoretical Computer Science}, 10, 2008.

\bibitem{chor1998private}
Benny Chor, Eyal Kushilevitz, Oded Goldreich, and Madhu Sudan.
\newblock Private information retrieval.
\newblock {\em Journal of the ACM}, 45(6):965--981, 1998.

\bibitem{correia2023short}
David~Munh Correia, Alexey Pokrovskiy, and Benny Sudakov.
\newblock Short proofs of rainbow matchings results.
\newblock {\em International Mathematics Research Notices}, 2023(14):12441--12476, 2023.

\bibitem{fazeli2015pir}
Arman Fazeli, Alexander Vardy, and Eitan Yaakobi.
\newblock {PIR} with low storage overhead: Coding instead of replication.
\newblock {\em arXiv preprint: 1505.06241}, 2015.

\bibitem{gao2021full}
Pu~Gao, Reshma Ramadurai, Ian~M Wanless, and Nick Wormald.
\newblock Full rainbow matchings in graphs and hypergraphs.
\newblock {\em Combinatorics, Probability and Computing}, 30(5):762--780, 2021.

\bibitem{hall1952combinatorial}
Marshall Hall.
\newblock A combinatorial problem on abelian groups.
\newblock {\em Proceedings of the American Mathematical Society}, 3(4):584--587, 1952.

\bibitem{hollmann2023some}
Henk~DL Hollmann, Karan Khathuria, Ago-Erik Riet, and Vitaly Skachek.
\newblock On some batch code properties of the simplex code.
\newblock {\em Designs, Codes and Cryptography}, 91(5):1595--1605, 2023.

\bibitem{ishai2004batch}
Yuval Ishai, Eyal Kushilevitz, Rafail Ostrovsky, and Amit Sahai.
\newblock Batch codes and their applications.
\newblock In {\em Proceedings of the thirty-sixth annual ACM symposium on Theory of computing}, pages 262--271, 2004.

\bibitem{karasev2012partitions}
Roman~N Karasev and Fedor~V Petrov.
\newblock Partitions of nonzero elements of a finite field into pairs.
\newblock {\em Israel Journal of Mathematics}, 192(1):143--156, 2012.

\bibitem{kohen2016generalization}
Daniel Kohen and Iv{\'a}n~Sadofschi Costa.
\newblock On a generalization of the seating couples problem.
\newblock {\em Discrete Mathematics}, 339(12):3017--3019, 2016.

\bibitem{kohen2010new}
Daniel Kohen and Ivan Sadofschi.
\newblock A new approach on the seating couples problem.
\newblock {\em arXiv preprint: 1006.2571}, 2010.

\bibitem{kovacs2023finding}
Benedek Kov{\'a}cs.
\newblock Finding a perfect matching of {$\mathbb{F}_2^n$} with prescribed differences.
\newblock {\em arXiv preprint: 2310.17433}, 2023.

\bibitem{lember2024equal}
J{\"u}ri Lember and Ago-Erik Riet.
\newblock Equal requests are asymptotically hardest for data recovery.
\newblock In {\em 2024 IEEE International Symposium on Information Theory (ISIT)}, pages 3678--3683, 2024.

\bibitem{nassar2021array}
Mohammad Nassar and Eitan Yaakobi.
\newblock Array codes for functional {PIR} and batch codes.
\newblock {\em IEEE Transactions on Information Theory}, 68(2):839--862, 2021.

\bibitem{preissmann2009seating}
Emmanuel Preissmann and Maurice Mischler.
\newblock Seating couples around the {K}ing's table and a new characterization of prime numbers.
\newblock {\em The American Mathematical Monthly}, 116(3):268--272, 2009.

\bibitem{yohananov2025optimal}
Lev Yohananov and Isaac~Barouch Essayag.
\newblock Optimal functional $2^{s-1}$-batch codes: Exploring new sufficient conditions.
\newblock In {\em 2025 IEEE International Symposium on Information Theory (ISIT)}, pages 1--6. IEEE, 2025.

\bibitem{yohananov2021almost}
Lev Yohananov and Eitan Yaakobi.
\newblock Almost optimal construction of functional batch codes using {H}adamard codes.
\newblock In {\em 2021 IEEE International Symposium on Information Theory (ISIT)}, pages 3139--3144. IEEE, 2021.

\bibitem{yohananov2022almost}
Lev Yohananov and Eitan Yaakobi.
\newblock Almost optimal construction of functional batch codes using extended simplex codes.
\newblock {\em IEEE Transactions on Information Theory}, 68(10):6434--6451, 2022.

\bibitem{zhang2020bounds}
Yiwei Zhang, Tuvi Etzion, and Eitan Yaakobi.
\newblock Bounds on the length of functional {PIR} and batch codes.
\newblock {\em IEEE Transactions on Information Theory}, 66(8):4917--4934, 2020.

\end{thebibliography}
\end{document}